\documentclass{ieeeconf}

\usepackage[usenames,dvipsnames]{xcolor}
\usepackage{quiver}
\usepackage{cancel}
\usepackage{amsfonts,amsmath,amssymb}
\usepackage[english]{babel}
\usepackage[normalem]{ulem}
\usepackage[hidelinks]{hyperref}

\usepackage{amsthm}
\usepackage{graphicx}
\usepackage{subcaption,caption}
\usepackage{epsfig}
\usepackage{ulem}
\usepackage{tikz}
\usepackage{mathrsfs}
\usepackage{bbold}             
\usepackage{bm}
\usetikzlibrary{shapes,arrows}
\usetikzlibrary{matrix}
\usetikzlibrary{calc,patterns,angles,quotes,arrows,automata}
\graphicspath{ {./imgs/}}

\usepackage[linesnumbered,noend]{algorithm2e}
\RestyleAlgo{ruled}
\SetKwInOut{Input}{Input}
\SetKwInOut{Output}{Output}
\SetKwInOut{Parameters}{Parameters}
\SetKwIF{If}{ElseIf}{Else}{if}{:}{else if}{else}{endif}

\theoremstyle{definition}

\theoremstyle{theorem}
\newtheorem{problem}{Problem}

\newtheorem{proposition}{Proposition}

\newtheorem{theorem}{Theorem}

\theoremstyle{definition}

\theoremstyle{remark}

\usepackage[backend=biber,
sorting=none,
style=numeric,
language=autobib,
doi=false,
url=false,
isbn=false,
firstinits=true
]{biblatex}
\newcommand{\Bc}{\mathcal{B}}

\newcommand{\Dc}{\mathcal{D}}

\newcommand{\Fc}{\mathcal{F}}

\newcommand{\Hc}{\mathcal{H}}
\newcommand{\Ic}{\mathcal{I}}
\newcommand{\Jc}{\mathcal{J}}
\newcommand{\Kc}{\mathcal{K}}
\newcommand{\Lc}{\mathcal{L}}

\newcommand{\Oc}{\mathcal{O}}
\newcommand{\Pc}{\mathcal{P}}

\newcommand{\Rc}{\mathcal{R}}

\newcommand{\Uc}{\mathcal{U}}

\newcommand{\As}{\mathscr{A}}

\newcommand{\Fs}{\mathscr{F}}

\newcommand{\Ls}{\mathscr{L}}

\newcommand{\Ns}{\mathscr{N}}
\newcommand{\Os}{\mathscr{O}}

\newcommand{\Vs}{\mathscr{V}}

\newcommand{\Cb}{\mathbb{C}}

\newcommand{\Nb}{\mathbb{N}}

\newcommand{\Rb}{\mathbb{R}}

\newcommand{\Tb}{\mathbb{T}}

\newcommand{\Bf}{\mathfrak{B}}

\newcommand{\Df}{\mathfrak{D}}

\newcommand{\rs}{\mathcal{R}}
\newcommand{\es}{\mathcal{J}}

\newcommand{\one}{\mathbb{1}}

\newcommand{\bin}{{\rm bin}}
\newcommand{\tr}{{\rm tr}}

\newcommand{\Span}{{\rm span}}

\newcommand{\alg}{{\rm alg}}

\newcommand{\ket}[1]{\left| #1 \right>}
\newcommand{\bra}[1]{\left< #1 \right|}

\newcommand{\inner}[2]{\left< #1,#2 \right>}

\newcommand{\expect}[1]{\left< #1 \right>}

\newcommand{\ketbra}[2]{\left\vert #1 \middle>\middle< #2 \right\vert}

\definecolor{darkviolet}{rgb}{0.58, 0.0, 0.83}
\definecolor{lavender}{rgb}{0.45, 0.31, 0.59}

\begin{document}

\title{Model Reduction for Controlled Quantum Markov Dynamics}

\author{Tommaso Grigoletto, Lorenza Viola and Francesco Ticozzi
\thanks{{T.G. and F.T. are with the Department of Information Engineering, and F.T. is also with the Padua Quantum Technologies Research Center, both at Università degli Studi di Padova, Italy.} L.V. is with the Department of Physics and Astronomy, Dartmouth College, Hanover, New Hampshire, 03755, USA, where F.T. also holds a secondary affiliation. F.T. acknowledges funding from the European Union - NextGenerationEU, within the National Center for HPC, Big Data and Quantum Computing (Project No.\,CN00000013, CN 1, Spoke 10). F.T and T.G. are partially supported by the Italian Ministry of University and Research under the PRIN project ``Extracting essential information and dynamics from complex networks", grant No.\,2022MBC2EZ. Work at Dartmouth was supported in part by the US Army Research Office under grant No.\,W911NF2210004.} 
}

\maketitle

\date{\today}

\begin{abstract}
We consider the problem of model reduction for Markovian quantum systems whose dynamics are described by a time-dependent Lindblad generator -- notably, as arising in the presence of external control. Our approach, which builds upon Krylov operator subspaces and operator-algebraic techniques introduced for time-independent generators, returns a reduced model that reproduces {\em exactly} the evolution of observables of interest and is {\em guaranteed} to be in Lindblad form. 
\end{abstract}

\section{Introduction}

Obtaining reduced-order models that respect structural constraints of dynamical models of interest is essential for reducing the demand of direct simulation and facilitate analysis and optimization tasks. 
We focus on a large class of quantum systems described by {\em time-dependent} Markovian (semigroup) dynamics \cite{lindblad1976generators,Gorini:1975nb} and leverage a series of results recently established for model reduction (MR) of time-independent generators \cite{tac2023,tit2025,quantum2025,ahp2025} to obtain {\em smaller-dimensional} descriptions of controlled, parametric or time-dependent quantum semigroups. Our operator-algebraic approach allows for {\em exactly} reproducing the expectation values of a set of target observables and, crucially, it guarantees that the reduced model is a valid quantum model, leading to completely-positive, trace-preserving (CPTP) reduced dynamics.

Physically, the kind of controlled dynamical models we consider arises in a variety of scenarios where time-local master equations are used to describe driven Markovian dynamics \cite{Kossakowski2010,Kolodynski2018,Landi2019} -- including the important case of periodically driven (Floquet-Markov) dynamics \cite{Mori2023} -- or where parametric families of generators are relevant -- such as in studies related to Hilbert space fragmentation \cite{HSF,Pollman2023}. In the latter case, treating these parameters as the controls, an exact MR for the set can be obtained using our approach. Another natural use of the proposed technique include the study of controlled open quantum systems: for instance, in situations where a fully-controllable system of interest interacts with an uncontrollable (possibly non-Markovian) bath, our techniques allow to obtain a simplified model in terms of a smaller-dimensional bath. 

Other applications may be envisioned in a variety of open-loop and feedback quantum-control scenarios \cite{altafini2012modeling,wiseman2009quantum}. In particular, in the context of feedback quantum control, the proposed technique can be used in combination with existing results \cite{ahp2025} to compute a reduced-order observer. This would lessen the computational effort associated with the real-time implementation of a filtering-based feedback scheme, potentially extending the viability of the techniques beyond what is achievable experimentally \cite{photon-box}. Our MR procedures are also well-suited for control strategies that entail switching between a finite number of generators \cite{ticozzi2012entanglement,scaramuzza2015,vom2019reachability, grigoletto2022, liang2024dissipative}.

The possibility of reducing a controlled dynamical system is naturally tied, and in fact complementary, to the extent the dynamical system explores the state space, and thus its controllability properties. Controllability analyses of Markovian open quantum systems have been developed under different assumptions \cite{Altafini_2003, Schirmer2010, Malvetti_2024}, and a general geometric approach based on Lie semigroups has been put forward \cite{d2021introduction,Dirr2009}.
While, as mentioned, our approach builds on techniques developed for time-independent generators \cite{quantum2025}, a key difference in the present analysis relates to the construction of the operator subspace that contains all the observables of interest evolved in the Heisenberg picture. 

The content is structured as follows. In Sec.\,\ref{sec:problem_setting} we introduce relevant notation and formulate the problem of interest. Sec.\,\ref{sec:quantum_model_reduction} contains our main contribution -- the  proposed observable-based MR algorithm. In Sec.\,\ref{sec:central_spin} we illustrate the MR procedure by studying a controlled dephasing central-spin model, along with mentioning other possible examples. 

\section{Problem setting: Controlled dynamics \\ in the Heisenberg picture}
\label{sec:problem_setting}

We consider finite-dimensional quantum systems, with associated Hilbert space $\Hc\simeq\Cb^n$, and operator space $\Bf(\Hc)\simeq\Cb^{n\times n}$. The state of the system is described by a density operator $\rho\in\Df(\Hc)$, where $\Df(\Hc) \equiv \{\rho\in\Bf(\Hc)| \rho=\rho^\dag\geq0, \tr(\rho)=1\}$, while physical observables are modelled as self-adjoint operators $O\in\Bf(\Hc)$, $O=O^\dag$. The expectation value of $O$ with respect to a state $\rho$ is given by $\expect{O}_\rho=\tr[O\rho]$. In the following, we shall drop the subscript $\rho$ whenever the state is clear from context. We describe the evolution of the system in the Heisenberg picture, via a time-local quantum master equation of the form
\begin{equation}
    \dot{O}(t) = \Lc_{u(t)}[O(t)].
    \label{eq:model}
\end{equation}
Here, $u(t):\Rb_{\geq 0}\to \Uc$ denotes the control input, that takes values in an admissible  set $\Uc\subseteq\Rb^m$, and $\Lc_{u}$ is a {\em controlled Lindblad generator}, namely, in units where $\hbar=1$, it may be expressed in the form \cite{Gorini:1975nb, lindblad1976generators}
\begin{equation}
\Lc_u  [O] = i [H_u, O] + \sum_k \Dc_{L_{u,k}} [O] ,
    \label{eq:lindblad}
\end{equation}
where $H_u=H_u^\dag$ describes the Hamiltonian (coherent) contribution to the dynamics, and the dissipative contributions have the form $\Dc_L[O] \equiv L^\dagger O L -\tfrac{1}{2} \{ L^\dagger L, O\}$, with $L_{u,k}\in\Bf(\Hc)$ being noise (Lindblad) operators. In the case where $\Lc$ is time-independent, super-operators of this type are the generators of CP and {\em unital} (i.e., identity-preserving) one-parameter quantum dynamical semigroups $\{\Lambda_{t,t_0} \equiv e^{\Lc (t-t_0)}\}_{t\geq t_0}$, resulting in CP-divisible (Markovian) dynamics \cite{benatti2024quantum}. 

The above general setting encompasses several representative situations of interest. In particular, we highlight a few:
\begin{itemize}
    \item Open-loop coherent control, whereby the dissipator is time-independent but $H(t) = H_0+ H_c(t)$, with $H_c(t) \equiv \sum_\ell u_\ell(t) H_\ell$, and $H_\ell\in \Bc(\Hc)$ represent independently tunable control directions \cite{dalessandro_quantum_2003}. 
    
    \item Dissipative control, whereby the strength of (some) noise operators is modulated in time, e.g., $L_{u,k}(t) \equiv u_k(t) \,L_k$, possibly in conjunction with time-dependent Hamiltonian control as above. Relevant settings include open-loop or feedback switching control scenarios \cite{ticozzi2012entanglement,scaramuzza2015, vom2019reachability, grigoletto2022, liang2024dissipative} or, in cases where the control inputs are time-periodic, Floquet-Lindblad master equations \cite{Mori2023}.
    
    \item Time-dependent perturbative scenarios, whereby the Hamiltonian and/or the noise operators are modified as $H(t)\equiv H_0 + u(t) H_1$,  $L_{u,k}(t) \equiv L_{0,k} + u_k(t) L_{1,k}$, such as in adiabatic elimination-type settings \cite{adiabadic-elimination}. 
\end{itemize}

Imagine we are given a finite (or finally generated) subset of observables of interest, say, $\Omega \equiv \{O_j\}$, whose time-dependent expectation values we wish to simulate more efficiently than using the full model. Without loss of generality, we can assume that $\one\in\Omega$. Given a state $\rho\in\Df(\Hc)$ and an initial observable $O(t_0)\in\Omega$ at time $t_0$, the solution of \eqref{eq:model} is formally given by the time-ordered exponential 
\begin{align}
   O(t) & \equiv \Lambda_{t,t_0}[O(t_0)] \equiv \Tb e^{\int_{t_0}^t \Lc_{u(\tau)}d\tau}[O(t_0)], \label{Tobs}\\
   &= \sum_{n=0}^{+\infty} \int_{t_0}^t dt_1 \dots \int_{t_0}^{t_{n-1}} dt_n \,\Lc_{u(t_1)} \dots \Lc_{u(t_n)}[O(t_0)], \nonumber
\end{align}
where $\Tb$ denotes the (Dyson) time-ordering symbol and, compared to the time-independent case, the maps $\{ \Lambda_{t, t_0}\}$ obey the more general forward composition law $\Lambda_{t,s} \circ \Lambda_{s,t_0} = \Lambda_{t,t_0}$, for any $t \geq s\geq t_0$ \cite{Kossakowski2010}. 
The resulting time-dependent expectation value is \(\expect{O(t)} = \tr[O(t)\rho]\). While care is needed in ensuring that time-dependent generators respect physical consistency conditions \cite{Kolodynski2018}, we will assume those to be obeyed and focus on mathematical aspects of the problem in what follows. 

By linearity of the dynamics, the ability to simulate the observables in $\Omega$ directly entails the ability to simulate any linear combinations of them. Hence, we can also consider $\Omega$ to be, in general, a finite-dimensional operator subspace. 

In this work, we focus on the following: 
\begin{problem}
    Given a time-dependent quantum model  \eqref{eq:model} and a finite set of target observables $\Omega$, find a smaller (if one exists) Hilbert space $\check{\Hc}$ on which we can define a controlled Lindblad generator $\check{\Lc}_{u}$, a time-dependent master equation 
    $$ \frac{d}{dt}\check{O}(t) = \check{\Lc}_{u(t)}[\check{O}(t)],$$
    and a map $\Phi:\Df(\Hc)\to\Df(\Kc)$ so that, for all $\rho\in\Df(\Hc)$, all $t\geq0$, and all $u(t):\Rb_{\geq0}\to\Uc$, we have 
    \[\expect{O(t)}_\rho = \expect{\check{O}(t)}_{\Phi(\rho)}.\]  
\end{problem}

As in the time-independent case \cite{quantum2025}, the key challenge is to ensure that the reduced generator $\check{\Lc}_{u}$ defines a \emph{valid quantum model} (indeed, a Lindbladian in our case). Satisfying this structural requirement is in fact a challenge for any quantum MR procedure -- notably, adiabatic elimination \cite{Tokieda_2024} -- and a prerequisite for implementation on a quantum simulator.

\section{Quantum model reduction for the observable dynamics}
\label{sec:quantum_model_reduction}
\subsection{Observable subspace}

Building on the results derived in \cite{quantum2025} for time-independent Markovian generators, we start by defining the notion of indistinguishable states: two states $\rho_1,\rho_2\in\Df(\Hc)$ are \textit{indistinguishable given the expectations in $\Omega$}, when for all control signals $u(t)$, all times and all observables $O\in\Omega$, we have $\expect{O(t)}_{\rho_1} = \expect{O(t)}_{\rho_2}$.
The \textit{(Krylov) observable subspace} is then defined as follows:
\begin{equation}
    \Os \equiv \Span\{\Tb e^{\int_0^t \Lc_{u(\tau)}d\tau}(O),\, \forall t\geq0, \forall u(t),\, \forall O\in\Omega\}.
    \label{eq:observable_subspace}
\end{equation}
The observable\footnote{The naming convention here is unfortunate, since it clashes between the physics and the control-systems theory literature, where the world ‘‘observable'' take different meanings.} subspace characterizes the states that are indistinguishable: by linearity of the dynamics, two states $\rho_1,\rho_2$ are indistinguishable (given the expectations in $\Omega$) when their difference is perpendicular to $\Os$, i.e., $\rho_1-\rho_2\perp\Os$ or $\tr[X(\rho_1-\rho_2)]=0$ for all $X\in\Os$.
Note that the operator space $\Os$ coincides with the attainable subspace in the Heisenberg picture. Through duality with respect to $\inner{\cdot}{\cdot}_{\text{HS}}$, one can also prove that $\Os$ coincides with the orthogonal subspace to the more common non-observable subspace, $\Ns$. The following proposition helps us to more explicitly characterize the observable subspace $\Os$:

\begin{proposition}
\label{prop:observable}
    $\Os$ is the smallest $\Lc_u$--invariant (for all $u\in\Cb^m$) subspace containing $\Omega$. 
    Furthermore, let $\Ls \equiv \alg\{\Lc_u\}$ denote the associative superoperator algebra generated by the controlled Lindblad generators. Then, we have \(\Os = \Span\{\Fc(O),\, \forall \Fc\in\Ls,\,\forall O\in\Omega\}.\)
\end{proposition}
\begin{proof}
The fact that $\Os$ contains $\Omega$ follows trivially from the definition, as $\Tb e^{\int_0^0 \Lc_{u(\tau)}d\tau} = \Ic$, the identity superoperator. We next prove that $\Os$ is $\Lc_u$--invariant for all $u\in\Cb^m$. Given an observable $O\in\Omega$, a time $t\geq t_0\equiv 0$, and a control $u(t)$, let $O(t) =\Tb e^{\int_0^{t} \Lc_{u(s)}ds}(O)$ as in Eq.\,\eqref{Tobs}. 
Then, for any $T\geq0$ and any $\bar{u}\in\Cb^m$, we can prove that $e^{\Lc_{\bar{u}}T}(O(t))\in\Os$.
Define a modified control function
\[u'(s) \equiv \begin{cases}
    u(s) &\text{if } s\leq t ,\\
    \bar{u} &\text{if } t<s\leq T.
\end{cases}\]
By using the forward composition property of the propagator, $\Lambda_{t,s} \circ \Lambda_{s,0} = \Lambda_{t,0},$ we have:
\begin{align*}
    \Tb e^{\int_0^{t+T} \Lc_{u'(s)}ds}(O)&=\Tb e^{\int_t^{T} \Lc_{u'(s)}ds}(e^{\int_0^{t} \Lc_{u'(s)}ds}O)\\&=e^{\Lc_{\bar{u}}T}(O(t)),
\end{align*}
where the first term is in $\Os$ by definition.
By linearity, since this holds for any generator of $\Os$, it holds for any operator in $\Os$, i.e., $e^{\Lc_{\bar{u}}T}[X]\in\Os$, for all $X\in\Os$. Equivalently,  we can rewrite the statement as follows: $\tr[e^{\Lc_{\bar{u}}T}(X)Y]=0$, for all $X\in\Os$, $Y\in\Os^\perp$, $\bar{u}$ and $T\geq0$. By expanding the exponential, we then have $\tr[e^{\Lc_{\bar{u}}T}(X)Y] = \sum_{k=0}^{+\infty} \frac{T^k}{k!} \tr[{\Lc_{\bar{u}}^k}(X)Y]$, which is equal to zero if and only if $\tr[{\Lc_{\bar{u}}^k}(X)Y]=0$ for all $k$. This implies that $\tr[{\Lc_{\bar{u}}}(X)Y]=0$, for all $X\in\Os$, $Y\in\Os^\perp$, $\bar{u}$ or, equivalently, ${\Lc_{\bar{u}}}(X)\in\Os$, for all $X\in\Os$ and $\bar{u}$, i.e., $\Os$ is $\Lc_u$--invariant, as claimed. 

To prove that $\Os$ is the smallest $\Lc_u$--invariant subspace that contains $\Omega$, consider an $\Lc_u$--invariant operator space $\Vs \supseteq \Omega$.
Then, for any $O\in\Omega$, $O(t)
\in \Vs$ since 
\begin{align*}
        O(t) = 
        \sum_{n=0}^{+\infty} \int_0^t dt_1 \dots \int_0^{t_{n-1}} dt_n \,\Lc_{u(t_1)} \dots \Lc_{u(t_n)}(O),
\end{align*}
and $\Lc_{u(t_1)} \dots \Lc_{u(t_n)}(O)\in\Vs$. Hence, $O(t)$
is the limit of a linear combination of operators in $\Vs$. This implies that $\Os\subseteq\Vs$ and thus $\Os$ is minimal.

Since $\Os$ is the smallest $\Lc_u$--invariant subspace that contains $\Omega$, we have $\Os = \Span\{\Lc_{u_1}\dots\Lc_{u_k}(O),\, \forall O\in\Omega\}$, for all sequences of values $u_1,\dots,u_k$ and the second part of the statement follows.
\end{proof}

A few points are worth emphasizing:

\noindent (i) Based on Proposition \ref{prop:observable}, it is necessary to compute the {\em associative} algebra $\alg\{\Lc_u\}$, rather than (as one may expect) the Lie algebra ${\rm Lie}\{\Lc_u\}$ \cite{elliott2009bilinear}, which is generally smaller, i.e. ${\rm Lie}\{\Lc_u\}\subseteq\alg\{\Lc_u\}$. This difference stems from the fact that, in our case, $\Os$ is the {\em operator space} that contains the observables evolved in the Heisenberg picture, whereas typically one is interested in computing the {\em set} of evolved observables. The fact that we are also interested in arbitrary linear combinations of evolved observables makes the associative algebra of superoperators $\Ls$ necessary. 

\noindent (ii) While $\{\Lc_u\}$ is not finite in general, in many cases it may not even be finitely-generated. Fortunately, in a few cases of interest, the set $\{\Lc_u\}$ can be taken to be finitely-generated. Those are the cases where $\Lc_u = \sum_{k=1}^m u_k(t) \Lc_{k}$, and include the cases of coherent and dissipative control we mentioned above (either modulated or switching). 

\noindent (iii) In practice, to compute $\Ls$, one can use the same procedure used to compute Lie  \cite{d2021introduction} and replace the Lie-bracket (commutator) $[\cdot,\cdot]$ with the associative (matrix) product. However, computing $\Ls$ is numerically demanding as it requires to compute products of superoperators (whose complexity scales as $\Oc(n^6)$), instead of products of operators (whose complexity scales instead as $\Oc(n^3)$). 
Notably, in the case where $\Lc_u$ is a parametric family of time-independent generators, the observable subspace can be more easily computed as
\(\Tilde{\Os} = \Span \left(\bigcup_u \{\Lc_u^k(O),\,O\in\Omega,k\in\Nb\}\right)\). It is not difficult to verify that $\Tilde{\Os}\subseteq\Os$.

\subsection{Reduced controlled quantum model}

By construction, the operator space $\Os$ contains all the observables of interest, evolved in the Heisenberg picture, for all the possible choices of control inputs $u(t)$: as such, it contains all the necessary degrees of freedom necessary to reproduce the evolution of the desired expectation values. If one were to project the original model of Eq.\,\eqref{eq:model} onto $\Os$, one would obtain the {\em minimal linear} reduced model capable of reproducing all the necessary expectation values $\expect{O(t)}$, similar to the time-independent case discussed in \cite{quantum2025}. Such a model, would not necessarily be a quantum model, however. In order to ensure that this constraint is met, we leverage the properties of operator algebras and CPTP projectors.

Specifically, we close the observable space $\Os$ to an operator algebra $\As\equiv \alg\{\Os\}$. By definition of $\alg$, $\As$ is the smallest associative algebra containing $\Os$ and, because we assumed $\one\in\Omega$, it follows that $\As$ is unital, i.e. $\one\in\As$. Let 
\begin{equation*}
    \As \equiv U\bigg(\bigoplus_k \Bf(\Hc_{F,k})\otimes \one_{G,k}\bigg)U^\dag
\end{equation*}
be the {\em Wedderburn decomposition} of $\As$ \cite{arveson,KLV,quantum2025}, where $\Hc = \bigoplus_k \Hc_{F,k} \otimes \Hc_{G,k}$ and $\one_{G,k}\in\Bf(\Hc_{G,k})$.
With this, we can define the CP and unital projector $\Pc:\Bf(\Hc)\to\Bf(\Hc)$ with ${\rm Im}\Pc = \As$ as 
\begin{equation}
    \Pc(X) \equiv U\bigg[\bigoplus_{k} \frac{\tr_{\Hc_{G,k}}(W_k X W_k^\dag)}{\dim(\Hc_{G,k})}\otimes \one_{G,k}\bigg]U^\dag,
    \label{proj}
\end{equation}
where $W_k$ is a linear operator from $\Hc$ to $\Hc_{F,k}\otimes\Hc_{G,k}$, such that $W_k W_k^\dag = \one_{F,k}\otimes \one_{G,k}$.

Note that $\As$ is isomorphic to a smaller algebra where all the repeated copies of the blocks induced by $\otimes\one_{G,k}$ are removed, that is $\check{\As} \equiv \bigoplus_k \Bf(\Hc_{F,k})\subseteq \Bf(\check{\Hc})$, where $\check{\Hc} = \bigoplus_k \Hc_{F,k}$. Then the projector superoperator $\Pc$ can be factorized into two CP and unital maps, $\Rc:\Bf(\Hc)\to\Bf(\check{\Hc})$, with ${\rm Im}\Rc = \check{\As}$, and $\Jc:\Bf(\check{\Hc})\to\Bf(\Hc)$, ${\rm Im}\Jc = \As$:
\begin{align*}
    \Rc(X) &= \bigoplus_k \frac{\tr_{\Hc_{G,k}}(W_k X W_k^\dag)}{\dim(\Hc_{G,k})} ,\\
    \Jc(\check{X}) &= \bigoplus_k X_k \otimes \one_{G,k},
\end{align*}
where $\check{X} \equiv \bigoplus_k X_k$ with $X_k\in\Bf(\Hc_{F,k}).$ One can easily verify that $\Pc,\Rc$ and $\Jc$ are indeed CP and unital.
These two maps then allow us to compute a reduced model:

\begin{proposition}
    Let $\Os$ be the observable subspace as defined in Eq.\eqref{eq:observable_subspace}, and let $\Jc$ and $\Rc$ be the CP and unital factors of the projector $\Pc$ onto $\As = \alg\{\Os\}$, Eq.\,\eqref{proj}. Let us further define the reduced observables as $\check{O} \equiv \Rc(O)$, for all $O\in\Omega$ and $\check{\Lc}_u = \Jc \Lc_u \Rc$. Then, for any state $\rho$, any control input $u(t)$ and any time $t\geq 0$, we have
    \[\expect{O(t)}_\rho = \expect{\check{O}(t)}_{\Jc^\dag(\rho)}, \;
    \text{ with }\;
    \check{O}(t) = \Tb e^{\int_0^t \check{\Lc}_{u(s)} ds}(\check{O}).\] 
\end{proposition}
\begin{proof}
First, let $\Pi$ denote a projector onto $\Os$. Then
\begin{align*}
    O(t) &= \Tb e^{\int_0^t \Lc_{u(s)} ds}(O) = \Pi \Tb e^{\int_0^t \Lc_{u(s)} ds} \Pi(O) \\ &= \Pi \Tb e^{\int_0^t \Pi\Lc_{u(s)}\Pi ds} \Pi(O),
\end{align*}
since $\Os$ is $\Lc_u$-invariant and thus $\Lc_u\Pi = \Pi\Lc_u\Pi$. 

On the other hand, we also have 
\begin{align*}
    \expect{\check{O}(t)}_{\Jc^\dag(\rho)} &= \tr[\Jc^\dag(\rho) \Tb e^{\int_0^t \check{\Lc}_{u(s)} ds}(\check{O})] \\& = \tr[\rho \Jc \Tb e^{\int_0^t \Rc\Lc_{u(s)}\Jc ds}\Rc(O)].
\end{align*}
Then, since $\As \supseteq \Os$, we have $\Pc\Pi = \Pi$ and thus 
\begin{align*}
    \Jc \Tb e^{\int_0^t \Rc\Lc_{u(s)}\Jc ds}\Rc(O) &= \Pc \Tb e^{\int_0^t \Pc\Lc_{u(s)}\Pc ds}\Pc(O) \\ &=\Pi \Tb e^{\int_0^t \Pi\Lc_{u(s)}\Pi ds}\Pi(O) , 
\end{align*}
which concludes the proof. 
\end{proof}

Furthermore, we can prove that the reduced controlled generator $\check{\Lc}_u$ is of Lindblad form: 
\begin{theorem}
\label{thm:reduced_Lindblad}
Let $\mathscr{A}$ be a unital $*$-subalgebra of $\mathcal{B(H)}, $ and let $\rs$ and $\es$ denote the CP and unital factors of the projector in \eqref{proj}, $\Pc = \es\rs$. Then, for any Lindblad generator $\Lc$, its reduction to $\As$, $\check{\Lc}\equiv\rs\Lc\es,$ is also a Lindblad generator, that is, $\check{\Lc}:\Bf(\check{\Hc})\to\Bf(\check{\Hc})$ and $\{e^{\check{\Lc}t}\}_{t\geq 0}$ is a CP, unital quantum dynamical semigroup. 
\end{theorem}

This result may be seen as the Heisenberg-picture version of an analogous result, Theorem 4, established in \cite{quantum2025}. While the latter theorem was proven for Lindblad generators of CPTP semigroups and CPTP projectors, the above modified version can be identically proven by using the duality of the Hilbert-Schmidt inner product $\inner{\cdot}{\cdot}_{HS}$. While Theorem \ref{thm:reduced_Lindblad} establishes that the reduced generator $\check{\Lc}_u$ is in Lindblad form, finding the reduced Hamiltonian $\check{H}_u \in\Bf(\check{\Hc})$ and noise operators $\check{L}_{u,k}\in\Bf(\check{\Hc})$ is non-trivial in general. We refer the interested reader to \cite[Appendix]{ahp2025} for how to compute the reduced operators. 

Finally, note that while projecting a Lindblad generator $\Lc$ onto an operator algebra is a {\em sufficient} condition to ensure that the reduced generator is in Lindblad form, it is {\em not} necessary in general. Nonetheless, the fact that $\As$ is the smallest algebra that contains $\Os$ ensures that the generator $\check{\Lc}_u$ is the smallest we can obtain with this MR procedure.

\subsection{Sufficient conditions for quantum model reduction}

As we mentioned, computing $\Ls$ (and thus $\Os$) is numerically demanding. For this reason, we next provide two easily verifiable sufficient conditions that allow us to check whether the model can be reduced.  

\begin{proposition}
\label{prop:suff_crit_1}
Let $\{H_u\}$ and $\{L_u\}$ be the controlled Hamiltonian and noise operators associated to the controlled Lindblad operator $\Lc_u$. Then,
\[\Os \subseteq \alg(\Os) \subseteq \alg(\{H_u\}\cup\{L_u\}\cup\{O_j\})\equiv \Fs.\]
\end{proposition}
\begin{proof}
    Because $\Fs$ is closed w.r.t. linear combinations and matrix products, and given the form \eqref{eq:lindblad} of $\Lc_u$, we have that $\Lc_u(F)\in\Fs$ for all $u\in\Uc$ and for all $F\in\Fs$, i.e., $\Fs$ is $\Lc_u$-invariant for all $u\in\Uc$. $\Omega$ is trivially contained in $\Fs$ and thus, since $\Os$ is the smallest $\Lc_u$-invariant subspace that contains $\Omega$, it must be $\Os\subseteq \Fs$. To conclude, we can observe that, by definition, $\alg(\Os)$ is the smallest algebra containing $\Os$, hence $\Os\subseteq\alg(\Os)\subseteq\Fs$. 
\end{proof}
 
Proposition \ref{prop:suff_crit_1} provides a sufficient criterion to verify if the original controlled model is reducible since, if $\Fs\subsetneq\Bf(\Hc)$, then we can certainly reduce it at least as much as $\Fs$ (possibly more). Since this is only a sufficient condition, however, it is possible to have reduction even if $\Fs=\Bf(\Hc)$. Also, while $\Fs$ is an $\Lc_u$-invariant algebra that contains $\Omega$, it need not be the minimal one. Still, in terms of numerical complexity, an important advantage resulting from this condition is that $\Fs$ is an operator algebra, and thus its calculation involves only products of operators (with a complexity of $O(n^3)$), not of superoperators. 

\begin{proposition}
\label{prop:suff_crit_2}
    Assume that the controlled Lindblad generator has the form $\Lc_u = \Lc_0 + \Lc_u',$ and let $\Os_0$ denote the observable space generated by $\Lc_0$ alone, i.e., \(\Os_0 \equiv \Span\{\Lc_0^i(O),\, \forall i\in\Nb,\, \forall O\in\Omega\}\). Then, if $\Os_0$ is $\Lc_u'$-invariant for all $u$, $\alg(\Os) = \alg(\Os_0)$.
\end{proposition}
\begin{proof}
    By definition, $\Os_0$ is the smallest $\Lc_0$--invariant subspace containing $\Omega$. By hypothesis, $\Os_0$ is also $\Lc_u'$--invariant, hence $\Os = \Os_0$ and the statement follows trivially. 
\end{proof}

The above result is especially relevant to settings where a generator $\check{\Lc_0}$ has been computed for the free dynamics \cite{quantum2025}; if $\Oc_0$ is preserved under the controlled dynamics, MR need only be carried out for $\Lc'_u$.  Note that, whenever the assumptions of Proposition \ref{prop:suff_crit_2} apply, it is numerically convenient to compute $\Os_0$ rather than $\Os$, since this calculation involves only repeated applications of the generator $\Lc_0$ \cite{quantum2025}. A characterization of the generators $\Lc$ that leave an operator algebra invariant is given in \cite{hasenohrl2023generators}.

\section{Controlled dephasing central-spin model} 
\label{sec:central_spin}

Consider a central-spin model inspired by studies of quantum decoherence \cite{zurek1982environment}, composed of one central spin-$1/2$ (a qubit), with $\Hc_S\simeq\Cb^2$ (denoted by $0$), and a bath made of $N$ spins (denoted by $k=1,\dots,N$), that is, $\Hc_B\simeq\Cb^{2^{N}}$ and $\Hc = \Hc_S\otimes \Hc_B$. We take the free (uncontrolled) evolution to be defined by the Hamiltonian $H_0 \equiv H_{B,0} + H_{{\rm int},0}$, with 
\begin{align*}
    H_{B,0} &= \sum_{k=1}^{N} J_{k,k}\sigma_z^{(k)} + \sum_{j,k=1}^{N} J_{j,k}\sigma_z^{(j)}\sigma_z^{(k)} ,\\
    H_{{\rm int},0} &= \sum_{k=1}^{N} J_{0,k}\sigma_z^{(0)}\sigma_z^{(k)},
\end{align*}
along with the noise operators $L_k = \gamma_k\,\sigma_z^{(k)}$, $\gamma_k \geq 0$, for all $k=1,\dots,N$. Here, $\sigma_\ell^{(j)}$, $\ell=0,x,y,z$, denotes a Pauli matrix acting non-trivially only on the $j$th-spin. We then assume to have full controllability over the central spin\footnote{While here the central spin is taken to be driftless, $H_{S,0}\equiv 0$, we could have equally assumed $H_{S,0}$ to be nonzero, e.g., $H_{S,0}\equiv \omega\sigma_z^{(0)}$, and complete controllability to arise from single-axis control, e.g., $u_1(t)\equiv 0$ in the above $H_c(t)$, as more common in physical settings.}, for instance via the control Hamiltonian $H_c (t) \equiv \sigma_x^{(0)}u_0(t) + \sigma_z^{(0)} u_1(t)$. Finally, we assume that the observables of interest are arbitrary central-spin observables, i.e., $O_\ell = \sigma_\ell^{(0)}$ for $\ell=0,x,y,z$, and $\Omega \simeq (\Cb^2)^2$.

Given the definitions above (and the fact that every bath operator is diagonal in the standard $z$ basis), one finds that
\begin{align*}
    \Fs&=\alg(\{H_0,H_u\}\cup\{L_k\}\cup\{O_j\})\\ &= \Span\{\sigma_u \otimes \ketbra{j}{j},\,u=0,x,y,z,\,j=0,\dots,2^{N}-1\}\\
    &\simeq \bigoplus_{q=1}^{2^{N}} \Bf(\Hc_S) ,
\end{align*}
where $\ket{j}$, with $j=0,\dots,2^{N}-1$, denotes the standard basis for $\Hc_B$. In the following, we take $\As=\Fs$ even tough, for certain choices of the parameters $J_{j,k}$ and $\gamma_k$, further MR could arise by joining together block-diagonal elements. This fact implies that, in this central spin model, the bath can be reduced to a classical model, making the overall system-bath pair an hybrid quantum-classical model \cite{barchielli2024hybridquantumclassicalsystemsquasifree}. 
The unitary matrix that takes $\As$ into its Wedderburn decomposition is the permutation matrix that inverts the order of the Kronecker product, i.e., $U\in\Bf(\Hc)$ such that $U(A\otimes B)U^\dag = B\otimes A$, for all $A\in\Bf(\Hc_S)$ and all $B\in\Bf(\Hc_B)$. 
In this case, the CP, unital projector onto $\As$ is given by 
\begin{align*}
    \Pc(X) = \bigoplus_{q=1}^{2^{N}} W_q X W_q^\dag ,\quad W_k \equiv (\bra{q}\otimes \one_2) U,
\end{align*}
with $\bra{q}$ denoting the standard basis for $\Cb^N$. The two factors can be taken to be the projector itself, i.e., $\Rc = \Jc = \Pc$. Note that in this case, even tough we do not obtain a reduction in the representation of the involved operators, since $\Rc(X)\in\Cb^{n\times n}$, there is still an effective MR as all the off-diagonal blocks are not relevant for the dynamics of the observables in $\Omega$. In the following, we denote with $\check{X}_q$ the operator in the $q$th block of the diagonal, i.e., $\check{X}_q \equiv W_q X W_q^\dag$.

In the new basis, the observables of interest take the form 
\[ U \sigma_u^{(0)} U^\dag = \bigoplus_{q=1}^{2^{N}} \sigma_u\]
and, similarly, the Hamiltonian and noise operators read 
\begin{align*}
    U H_c(t) U^\dag &= \bigoplus_{q=1}^{2^{N}} \sigma_x u_0(t) + \sigma_z u_1(t),\\
    U H_{B,0} U^\dag &= \bigoplus_{q=1}^{2^{N}} \alpha_q \one_2,\\
    U H_{{\rm int},0} U^\dag &= \bigoplus_{q=1}^{2^{N}} \beta_{q} \sigma_z,\\
    U L_k U^\dag &= \bigoplus_{q=1}^{2^{N}} \eta_{k,q} \gamma_k\one_2.
\end{align*}
Here, the coefficients $\alpha_q\in\Rb$ depend on the couplings $J_{j,k}$, $\beta_q\equiv \sum_{k=0}^{N} (-1)^{\bin(q)[k]} J_{0,k}$, where $\bin(q)[k]\in\{0,1\}$ denotes the $k$-th bit in the binary representation of $q$ (in lexicographical order, i.e., $k=0$ denotes the least significant bit),
and $\eta_{k,q}\in\{\pm1\}$ depending on $k$ and $q$.

For any initial state $\rho\in\Df(\Hc)$ (whether factorized or entangled across system and bath), the reduced state is then given by $\Jc^\dag(\rho) = \bigoplus_{q=1}^{2^N} W_q \rho W_q^\dag$, where only the blocks in the diagonal $\check{\rho}_q$ are relevant. Furthermore, given the structure of $H_{B,0}$ and $L_k$, the reduced model turns out to be equivalent to a \emph{classical} ensemble of $2^{N}$ spins, each one evolving with its own dynamics given by 
\begin{align}
   \dot{\check{O}}_q(t) &= -i[\check{H}_q(t),\check{O}_q(t)],
   \label{eq:reduced_model}\\
   \check{H}_q(t) &= [\beta_q + u_1(t)]\sigma_z + u_0(t)\sigma_x.\nonumber
\end{align}
The expectation value of any observable of interest $\sigma_\ell^{(0)}$ is then given by 
\[ \expect{\sigma_\ell^{(0)}(t)}_\rho = \sum_{q=1}^{2^N} \expect{\sigma_{\ell,q}(t)}_{\check{\rho}_q}\]
where $\sigma_{\ell,q}(t)$ is the solution of Eq.\,\eqref{eq:reduced_model} when $\check{O}_q(0) = \sigma_\ell$. 

If we assume to introduce a controlled dissipation to the model through some noise operators, say, $L_u (t) = u_2(t) \sigma_x^{(k)}$ for some $k=1,\dots,N$, then $\alg(\{H_0,H_u\}\cup\{L_k,L_u\}\cup\{O_j\})$ is no longer diagonal in the bath basis. Nevertheless, one may note that $\As$ is left invariant by the generator $\Dc_{L_u}(O)$ and thus, by a trivial extension of Proposition \ref{prop:suff_crit_2}, the model can still be reduced by projecting it onto $\As$. By direct calculations, one can also observe that, in this case, the reduced model remains the one described by Eq.\,\eqref{eq:reduced_model} since the reduced controlled operator acts trivially on the observables of interest. The same reasoning also hods for other local controlled dissipations acting on the bath along different axes, e.g., $\sigma_y,\sigma_+,\sigma_-$, or even for collective dissipation, e.g., $L_u (t) = u_2(t) \sum_{k=1}^N \sigma_x^{(k)}$ or similar.

\medskip

\section{Conclusions and Outlook }

In this work we extended the results of \cite{quantum2025} to the case of time-dependent quantum Markov dynamics, for which we are only interested in reproducing the evolution of a set of observables of interest. The proposed MR reduction method returns a smaller-dimensional \emph{quantum} model that exactly reproduces the expectation value of the observables of interest and is still in Lindblad form.

While, in terms of applications, we have focused here on a particularly simple model, for which the MR can be carried out analytically and which results in effectively classical Markov dynamics, several more complex applications may be envisioned. For instance, building on MR for time-independent generators studied in \cite{quantum2025}, the effect of arbitrary Hamiltonian control on a central spin may be examined in more complex dissipative central-spin models, such as collectively-coupled XYZ models. Likewise, the introduction of time-dependent boundary dissipation in XXZ spin chains, for instance by time-periodic or switched control of boundary Lindblad operators, could be of interest for non-equilibrium many-body physics.

Finally, although in this work we only focused on observable-based MR, the proposed procedure can be adapted to cases where the set of initial conditions is restricted to a few states of interest, thus obtaining a reachable-based MR \cite{quantum2025}. In this case, the optimal MR may need to be performed on an operator algebra with respect to a modified product, a so-called ``distorted algebra''. 

\medskip

\textit{Acknowledgment}: The authors would like to thank Yukuan Tao for valuable discussions on the topics of this work.

\printbibliography{}

\end{document}